\documentclass[11pt,letterpaper]{article}
\usepackage{times,mathptmx}
\DeclareMathAlphabet{\mathcal}{OMS}{cmsy}{m}{n}
\usepackage{fullpage}
\usepackage{amsmath,amsfonts,amssymb,amsthm,amscd}
\usepackage{tabularx,booktabs}
\usepackage{hyperref,cite,url}
\hypersetup{pdfpagemode=UseNone,pdfstartview=}
\theoremstyle{plain} 
\newtheorem{theorem}{Theorem}[section]
\newtheorem{lemma}[theorem]{Lemma}

\newtheorem{corollary}[theorem]{Corollary}

\theoremstyle{definition} 
\newtheorem{definition}{Definition}[section]
\newtheorem{assumption}{Assumption}
\theoremstyle{remark} 

\newcommand{\svs}{\vspace{0.7mm}}
\newcommand{\vs}{\vspace{1.5mm}}
\newcommand{\G}{\mathbb{G}}
\newcommand{\Z}{\mathbb{Z}}
\newcommand{\bits}{\{0,1\}}
\newcommand{\mc}[1]{\mathcal{#1}}
\newcommand{\tb}[1]{\textbf{#1}}

\newcommand{\lb}{\linebreak[0]}

\newcommand{\Adv}{\textbf{Adv}}

\title{Adaptively Secure Distributed Broadcast Encryption\\ with Linear-Size
Public Parameters}

\author{
    Kwangsu Lee\footnote{Sejong University, Seoul, Korea. Email:
	\texttt{kwangsu@sejong.ac.kr}.}
}

\date{}

\begin{document}

\maketitle

\begin{abstract}
Distributed broadcast encryption (DBE) is a variant of broadcast encryption
(BE) that can efficiently transmit a message to a subset of users, in which
users independently generate user private keys and user public keys instead of
a central trusted authority generating user keys. In this paper, we propose a
DBE scheme with constant size ciphertexts, constant size private keys, and
linear size public parameters, and prove the adaptive security of our DBE
scheme under static assumptions in composite-order bilinear groups. The
previous efficient DBE schemes with constant size ciphertexts and constant size
private keys are proven secure under the $q$-Type assumption or have a drawback
of having quadratic size public parameters. In contrast, our DBE scheme is the
first DBE scheme with linear size public parameters proven adaptively secure
under static assumptions in composite-order bilinear groups. 
\end{abstract}

\vs \noindent {\bf Keywords:} Broadcast encryption, Distributed broadcast
encryption, Adaptive security, Bilinear maps.

\newpage

\section{Introduction}

Broadcast encryption (BE) is a special kind of an encryption mechanism in
which a ciphertext is associated with a set of recipients, and a user
belonging to the set of recipients can decrypt the ciphertext with their own
private key \cite{FiatN93}. A non-trivial BE scheme must have sublinear size
ciphertexts since a trivial BE scheme with linear size ciphertexts can be
easily constructed by simply concatenating ciphertexts of public-key
encryption (PKE). Many public-key BE schemes with constant size ciphertexts
that allow anyone to create a ciphertext have been proposed \cite{BonehGW05,
GentryW09,Waters09}. However, the biggest drawback of existing efficient BE
schemes is that a central trusted authority is required to generate private
keys of users. In PKE, a central trusted authority is not needed for key
generation because individual users independently generate private keys and
public keys. The need for a central trusted authority is an obstacle that
hinders the application of BE schemes to decentralized environments that have
recently been attracting attention, such as blockchains.

Distributed broadcast encryption (DBE) is a variant of BE in which users can
independently generate their own private and public keys, and there is no need
of a central trusted authority for key generation \cite{WuQZD10,BonehZ14}. As
a result, the encryption and decryption algorithms of DBE require the public
keys of recipients in addition to public parameters, which increases the
storage size for storing the public keys of users. Previously, many DBE
schemes have been proposed by using bilinear pairing, indistinguishability
obfuscation, and lattices \cite{WuQZD10,BonehZ14,KolonelosMW23,ChampionW24}.
In reality, the most efficient DBE schemes are those designed in pairing
groups, which have $O(1)$ ciphertext size, $O(1)$ user private key size,
$O(L)$ user public key size, and $O(L)$ or $O(L^2)$ public parameters size
where $L$ is the number of users \cite{KolonelosMW23}. In addition, these
efficient DBE schemes provide adaptive security that allows an attacker to
select the target set for attacks in the challenge phase.

The most efficient DBE scheme is the KMW-DBE scheme with linear size public
parameters which is derived from the BGW-BE scheme in prime-order bilinear
groups, which has been proven to be adaptive secure under the $q$-Type
assumption \cite{KolonelosMW23}. In general, the $q$-Type assumption has been
widely used to prove the security of efficient BE schemes, but it has a
disadvantage that the security decreases as the parameter $q$ increases where
$q$ is dependent on the number of private keys. In this paper, we ask whether
it is possible to construct a DBE scheme with linear size public parameters
and prove the adaptive security under static assumptions instead of the $q$-Type
assumption.

\subsection{Our Contributions}

In this paper, we first propose an efficient DBE scheme with linear-size
public parameters and prove its semi-static security based on static
assumptions in composite-order bilinear groups. The semi-static security model
which was introduced by Gentry and Waters \cite{GentryW09} is weaker than the
adaptive security model where an attacker selects the target set in the
challenge phase, but stronger than the static security model where the
attacker must submit the target set in the initial phase \cite{BonehGW05}.
However, a semi-statically secure DBE scheme can be converted into an
adaptively secure DBE scheme by using the conversion method proposed by Gentry
and Waters \cite{GentryW09,KolonelosMW23}. Thus, we obtain the first DBE
scheme that has constant size ciphertexts, constant size user private keys,
linear size user public key, and linear-size public parameters, and prove the
adaptive security of our DBE scheme under static assumptions in
composite-order bilinear groups. The comparison of our DBE scheme with the
previous BE and DBE schemes is given in Table \ref{tab:comp-be-dbe}. 

The basic idea of designing a DBE scheme is to decentralize the private key
generation process of a BE scheme \cite{KolonelosMW23}. The most efficient BE
scheme is the BGW-BE scheme because it has linear size public parameters,
constant size private keys, and constant size ciphertexts, but the static
security of this scheme can be only proven under the $q$-Type assumption
\cite{BonehGW05}. The KMW-DBE scheme is an efficient DBE scheme derived from
the BGW-BE scheme by decentralizing the private key generation process, but
its semi-static security can be only proven under the $q$-Type assumption by
using the partitioning technique \cite{KolonelosMW23}. Our DBE scheme is a
modification of the KMW-DBE scheme to use composite-order bilinear groups
instead of prime-order bilinear groups. To prove the semi-static security of
our DBE scheme under static assumptions instead of the $q$-Type assumption, we
use the dual system encryption technique and its variant, which were widely
used in the proofs of existing identity-based encryption (IBE), hierarchical
IBE (HIBE), and attribute-based encryption (ABE) schemes
\cite{Waters09,LewkoW10,ChaseM14,Wee16}.

\begin{table*}[t]
\caption{Comparison of broadcast encryption schemes in bilinear groups}
\label{tab:comp-be-dbe}
\vs \small \addtolength{\tabcolsep}{8.2pt}
\renewcommand{\arraystretch}{1.4}
\newcommand{\otoprule}{\midrule[0.09em]}
\begin{tabularx}{6.50in}{lccccccc}
\toprule
Scheme  & Type & PP  & USK  & UPK  & CT  & Model & Assumption \\
\otoprule
BGW \cite{BonehGW05}
	& BE & $O(L)$ & $O(1)$ & - & $O(1)$     & ST & $q$-Type \\
AKN \cite{AbdallaKN07}
    & BE & $O(L)$ & $O(L)$ & - & $O(1)$     & ST & $q$-Type \\
GW \cite{GentryW09}
    & BE & $O(L)$ & $O(L)$ & - & $O(1)$     & AD & $q$-Type \\
GW \cite{GentryW09}
    & BE & $O(L)$ & $O(1)$ & - & $O(1)$     & AD & $q$-Type \\
Waters \cite{Waters09}
    & BE & $O(L)$ & $O(L)$ & - & $O(1)$     & AD & DLIN \\
Wee \cite{Wee16}
    & BE & $O(L)$ & $O(1)$ & - & $O(1)$     & ST & SD, GSD \\
GKW \cite{GayKW18}
    & BE & $O(L^2)$ & $O(1)$ & - & $O(1)$   & AD & SD, GSD \\
GKW \cite{GayKW18}
    & BE & $O(L^2)$ & $O(1)$ & - & $O(1)$   & AD & $k$-LIN \\
HWW \cite{HsiehWW25}
    & BE & $O(L)$ & $O(1)$ & - & $O(1)$  	& AD & $q$-Type \\
\midrule[0.03em]
WQZD \cite{WuQZD10}
    & DBE & $O(L)$ & $O(L)$ & $O(L^2)$ & $O(1)$	& AD & $q$-Type \\
KMW \cite{KolonelosMW23}
    & DBE & $O(L)$ & $O(1)$ & $O(L)$ & $O(1)$   & AD & $q$-Type \\
KMW \cite{KolonelosMW23}
    & DBE & $O(L^2)$ & $O(1)$ & $O(L)$ & $O(1)$ & AD & $k$-LIN \\
\midrule[0.05em]
Ours & DBE & $O(L)$ & $O(1)$ & $O(L)$ & $O(1)$   & AD & SD, GSD \\
\bottomrule
\multicolumn{8}{p{6.10in}}{
Let $L$ be the number of all users.
We count the number of group elements to measure the size.
We use symbols ST for static security and AD for adaptive security.
}
\end{tabularx}
\end{table*}

\subsection{Related Work}

\tb{Broadcast Encryption.} The concept of broadcast encryption (BE), which can
securely transmit a message to a subset of users, was introduced by Fiat and
Naor \cite{FiatN93}. Naor et al. proposed symmetric-key BE schemes by using
the subset cover framework and showed that their schemes provide collusion
resistance security \cite{NaorNL01}. The ciphertexts of symmetric-key BE can
only be created by a central trusted authority, but the ciphertexts of
public-key BE can be created by anyone. Boneh et al. proposed the first
public-key BE scheme which has constant size ciphertexts and constant size
private keys in prime-order bilinear groups and proved its static security
under the $q$-Type assumption \cite{BonehGW05}. Abdalla et al. showed that it
is possible to convert an HIBE scheme with the private key delegation property
into a BE scheme and proposed an efficient BE scheme with constant size
ciphertexts and linear size private keys based on the BBG-HIBE scheme
\cite{AbdallaKN07}. The ideal security model of BE is an adaptive security
model in which an attacker selects the target subsets in the challenge phase.
Gentry and Waters presented a conversion method that transform a
semi-statically secure BE scheme into an adaptively secure BE scheme and
proposed an adaptively secure identity-based BE scheme \cite{GentryW09}.
Waters proposed a BE scheme with constant size private keys and proved its
adaptive security by using the dual system encryption technique under the
standard assumption \cite{Waters09}. Gay et al. presented an efficient BE
scheme with square size public parameters, constant size ciphertexts, and
constant size private keys, and proved its adaptive security under standard
assumptions \cite{GayKW18}. Hsieh et al. presented another BE scheme which
reduces the square size public parameters of the GKW-BE scheme to linear size
public parameters by compressing the public parameters \cite{HsiehWW25}.

\svs\noindent \tb{Distributed Broadcast Encryption.} While BE requires a central
trusted authority that generates users' private keys, distributed BE (DBE)
does not require the central trusted authority since it allows users to
independently generate their own private and public keys \cite{WuQZD10,
BonehZ14}. Wu et al. proposed the concept of Ad hoc broadcast encryption
(AHBE) that does not require a central trusted authority and proposed an AHBE
scheme with relatively large size private and public keys \cite{WuQZD10}.
Boneh and Zhandry proposed the concept of DBE and showed that the most
efficient DBE scheme can be constructed by using indistinguishability
obfuscation \cite{BonehZ14}. Kolonelos et al. showed that it is possible to
convert existing BGW-BE and GKW-BE schemes in bilinear groups into DBE schemes
and proposed an efficient DBE scheme that provides the adaptive security using
DSE technique under the standard assumption \cite{KolonelosMW23}. The GW
transformation that converts a semi-statically secure BE scheme into an
adaptive secure BE scheme is equally applicable to DBE schemes
\cite{KolonelosMW23}. Recently, Champion and Wu proposed the first DBE scheme
based on lattices and proved its security under the modified LWE assumption
\cite{ChampionW24}. Garg et al. introduced the concept of flexible broadcast
encryption (FBE) that does not require to specify a user index when generating
a user secret key and proposed a conversion method to convert a DBE scheme
into an FBE scheme \cite{GargLWW23}. We may view DBE is a special case of
silent threshold encryption (STE) that supports distributed private key
generation, and Garg et al. proposed a secure and efficient STE scheme in the
generic group model and showed that an efficient DBE scheme can be derived
from their STE scheme \cite{GargKPW24}.

\svs\noindent \tb{Registration-Based Encryption.} Identity-based encryption
(IBE) is a variant of public-key encryption in which the public key of a user
is replaced by an identity string, and it requires a trusted authority to
generate a private key corresponding to the identity of a user
\cite{BonehF01}. Registration-based encryption (RBE) is an extension of IBE
that replaces the trusted authority with a key curator who simply registers
public keys of users without knowledge of any secret keys \cite{GargHMR18}.
Recently, the concept of registered attribute-based encryption (Reg-ABE) was
also introduced by applying RBE to attribute-based encryption (ABE), and an
efficient Reg-ABE scheme in bilinear groups was proposed
\cite{HohenbergerLWW23}. Since ABE can play the role of BE, a Reg-ABE scheme
can be naturally converted to a DBE scheme. Many Reg-ABE schemes have been
proposed in bilinear groups and lattices \cite{ZhuZGQ23,AttrapadungT24,
GargLWW24,ChampionHW25}.

\section{Preliminaries}

In this section, we define symmetric-key encryption, the bilinear groups of
composite-order, and complexity assumptions.

\subsection{Symmetric Key Encryption}

\begin{definition}[Symmetric Key Encryption]
A symmetric key encryption (SKE) scheme consists of three algorithms 
$\tb{GenKey}, \tb{Encrypt}$, and $\tb{Decrypt}$, which are defined as follows:
\begin{description}
\item $\tb{GenKey}(1^\lambda)$: The key generation algorithm takes as input a 
security parameter $\lambda$. It outputs a symmetric key $K$.

\item $\tb{Encrypt}(K, M)$: The encryption algorithm takes as input a 
symmetric key $K$ and a message $M$. It outputs a ciphertext $C$.

\item $\tb{Decrypt}(K, C)$: The decryption algorithm takes as input a
symmetric key $K$ and a ciphertext $C$. It outputs a message $M$ or a special
symbol $\perp$.
\end{description}
The correctness property of SKE is defined as follows: For all $K$ generated
by $\tb{GenKey}(1^\lambda)$ and any message $M$, it is required that
$\tb{Decrypt} (K, \tb{Encrypt}(K, M)) = M$. 
\end{definition}

\begin{definition}[One-Message Indistinguishability] \label{def:ske-omi-sec}
The one-message indistinguishability (OMI) of SKE is defined in terms of the 
following experiment between a challenger $\mc{C}$ and a PPT adversary $\mc{A}$ 
where $1^\lambda$ is given as input:
\begin{enumerate}
\item \tb{Setup}: $\mc{C}$ obtains a symmetric key $K$ by running $\tb{GenKey} 
(1^{\lambda})$ and keeps $K$ to itself.

\item \tb{Challenge}: $\mc{A}$ submits challenge messages $M_0^*, M_1^*$ where
$|M_0^*| = |M_1^*|$. $\mc{C}$ flips a random coin $\mu \in \bits$ and obtains
$CT^*$ by running $\tb{Encrypt}(K, M_\mu^*)$. It gives $CT^*$ to $\mc{A}$.

\item \tb{Guess}: $\mc{A}$ outputs a guess $\mu' \in \bits$. $\mc{C}$ outputs
$1$ if $\mu = \mu'$ or $0$ otherwise.
\end{enumerate}
The advantage of $\mc{A}$ is defined as $\Adv_{SKE,\mc{A}}^{OMI} (\lambda) =
\big| \Pr[\mu = \mu'] - \frac{1}{2} \big|$ where the probability is taken over
all the randomness of the experiment. An SKE scheme is OMI secure if for all
probabilistic polynomial-time (PPT) adversary $\mc{A}$, the advantage of
$\mc{A}$ is negligible in the security parameter $\lambda$.
\end{definition}

\subsection{Bilinear Groups of Composite Order}

Let $N=p_1 p_2 p_3$ where $p_1, p_2$, and $p_3$ are distinct prime numbers.
Let $\G$ and $\G_T$ be two multiplicative cyclic groups of same composite
order $N$ and $g$ be a generator of $\G$. The bilinear map $e:\G \times \G
\rightarrow \G_{T}$ has the following properties:
\begin{enumerate}
\item Bilinearity: $\forall u, v \in \G$ and $\forall a,b \in \Z_N$,
    $e(u^a,v^b)=e(u,v)^{ab}$.
\item Non-degeneracy: $\exists g$ such that $e(g,g)$ has order $N$, that
    is, $e(g,g)$ is a generator of $\G_T$.
\end{enumerate}
We say that $\G$ is a bilinear group if the group operations in $\G$ and
$\G_T$ as well as the bilinear map $e$ are all efficiently computable.
Furthermore, we assume that the description of $\G$ and $\G_T$ includes
generators of $\G$ and $\G_T$ respectively. We use the notation $\G_{p_i}$ to
denote the subgroups of order $p_i$ of $\G$ respectively. Similarly, we use
the notation $\G_{T,p_i}$ to denote the subgroups of order $p_i$ of $\G_T$
respectively.
We note that if $h_i \in \G_{p_i}$ and $h_j \in \G_{p_j}$ for $i \neq j$,
then $e(h_i, h_j)$ is the identity element in $\G_T$. This orthogonality
property of $\G_{p_1}, \G_{p_2}, \G_{p_3}$ will be used to implement
semi-functionality in our constructions.

\subsection{Complexity Assumptions}

\begin{assumption}[Subgroup Decision, SD]
Let $(N, \G, \G_T, e)$ be a description of the bilinear group of composite
order $N=p_1 p_2 p_3$. Let $g_1, g_2, g_3$ be generators of subgroups
$\G_{p_1}, \G_{p_2}, \G_{p_3}$ respectively. The SD assumption is that if the
challenge tuple
    $$D = ((N, \G, \G_T, e), g_1, g_3) \mbox{ and } Z$$
are given, no PPT algorithm $\mc{A}$ can distinguish $Z = Z_0 = X_1 \in
\G_{p_1}$ from $Z = Z_1 = X_1 R_1 \in \G_{p_1 p_2}$ with more than a
negligible advantage. The advantage of $\mc{A}$ is defined as
    $\Adv_{\mc{A}}^{SD}(\lambda) = \big| \Pr[\mc{A}(D,Z_0) = 0] -
    \Pr[\mc{A}(D,Z_1) = 0] \big|$
where the probability is taken over random choices of $X_1 \in \G_{p_1}$ and
$R_1 \in \G_{p_2}$.
\end{assumption}

\begin{assumption}[General Subgroup Decision, GSD]
Let $(N, \G, \G_T, e)$ be a description of the bilinear group of composite
order $N=p_1 p_2 p_3$. Let $g_1, g_2, g_3$ be generators of subgroups
$\G_{p_1}, \G_{p_2}, \G_{p_3}$ respectively. The GSD assumption is that if
the challenge tuple
    $$D = ((N, \G, \G_T, e), g_1, g_3, X_1 R_1, R_2 Y_1) \mbox{ and } Z$$
are given, no PPT algorithm $\mc{A}$ can distinguish $Z = Z_0 = X_2 Y_2 \in
\G_{p_1 p_3}$ from $Z = Z_1 = X_2 R_3 Y_2 \in \G_{p_1 p_2 p_3}$ with more
than a negligible advantage. The advantage of $\mc{B}$ is defined as
    $\Adv_{\mc{A}}^{GSD} (\lambda) = \big| \Pr[\mc{A}(D,Z_0) = 0] -
    \Pr[\mc{A}(D,Z_1) = 0] \big|$
where the probability is taken over random choices of $X_1, X_2 \in
\G_{p_1}$, $R_1, R_2, R_3 \in \G_{p_2}$, and $Y_1, Y_2 \in \G_{p_3}$.
\end{assumption}

\section{Distributed Broadcast Encryption}

In this section, we define the syntax of DBE and its security models.

\subsection{Definition}

In a DBE scheme, a trusted authority generates public parameters to be used in
the system by running the setup algorithm. Each user generates a user private
key and a user public key for the user’s index by running the key generation
algorithm with the public parameters as input and stores the user public key
in a public directory. Then, a sender creates a ciphertext for a subset of
users by running the encryption algorithm with the receivers' public keys and
public parameters as input. A receiver can decrypt the ciphertext using his
private key if its index belongs to the subset of the ciphertext. A more
detailed syntax of the DBE scheme is given as follows.

\begin{definition}[Distributed Broadcast Encryption]
A distributed broadcast encryption (DBE) scheme consists of five algorithms
$\tb{Setup}, \tb{GenKey}, \tb{IsValid}, \tb{Encaps}$, and $\tb{Decaps}$,
which are defined as follows:
\begin{description}
\item $\tb{Setup}(1^\lambda, 1^L)$: The setup algorithm takes as input a
security parameter $1^\lambda$, and the number users $L$. It outputs public
parameters $PP$.

\item $\tb{GenKey}(i, PP)$: The key generation algorithm takes as input a user
index $i \in [L]$ and public parameters $PP$. It outputs a private key $USK_i$
and a public key $UPK_i$.

\item $\tb{IsValid}(j, UPK_j, PP)$: The public key verification algorithm
takes as input an index $j$, a public key $UPK_j$, and the public parameters
$PP$. It outputs $1$ or $0$ depending on the validity of keys.

\item $\tb{Encaps}(S, \{ (j, UPK_j) \}_{j \in S}, PP)$: The encapsulation
algorithm takes as input a set $S \subseteq [L]$, public keys $\{ (j, UPK_j)
\}_{j \in S}$, and public parameters $PP$. It outputs a ciphertext header $CH$ 
and a session key $CK$.

\item $\tb{Decaps}(S, CH, i, USK_i, \{ (j, UPK_j) \}_{j \in S}, PP)$: The
decapsulation algorithm takes as input a set $S$, a ciphertext header $CH$, an
index $i$, a private key $USK_i$ for the index $i$, public keys $\{ (j, UPK_j)
\}_{j \in S}$, and public parameters $PP$. It outputs a session key $CK$ or
$\perp$.
\end{description}
The correctness of DBE is defined as follows: For all $PP$ generated by
$\tb{Setup}(1^{\lambda}, 1^L)$, all $(USK_i, UPK_i)$ generated by $\tb{GenKey}
(i, PP)$, all $UPK_j$ such that $\tb{IsValid}(j, UPK_j, PP)$, all $S \subseteq 
[L]$, it is required that
\begin{itemize}
\item If $i \in S$, then $CK = CK'$ where $(CH, CK) = \tb{Encaps}(S, \{ (j, 
UPK_j) \}_{j \in S}, PP)$ and $CK' = \tb{Decaps}(S, \lb CH, i, USK_i, \{ (j, 
UPK_j) \}_{j \in S}, PP)$.
\end{itemize}
\end{definition}

\subsection{Security Model}

The semi-static security model is an enhanced security model of the static
security model in which an attacker specifies the challenge set $S^*$ before
it sees the public parameters \cite{GentryW09}. In the semi-static security
model, an attacker first commits an initial set $\tilde{S}$, and a challenger
generates public parameter $PP$ and gives it to the attacker. Afterwards, the
attacker can obtain the public keys of users belonging to $\tilde{S}$. In the
challenge phase, the attacker submits the challenge set $S^*$, which is a
subset of $\tilde{S}$, and obtains the challenge ciphertext header $CH^*$ and
the challenge session key $CK_\mu^*$. Finally, the attacker succeeds if he can
guess whether the challenge session key is correct or random. The detailed
description of this security model is given as follows:

\begin{definition}[Semi-Static Security] \label{def:dbe-ss-sec}
The semi-static security of DBE is defined in terms of the following 
experiment between a challenger $\mc{C}$ and a PPT adversary $\mc{A}$ where
$1^\lambda$ and $1^L$ are given as input:
\begin{enumerate}
\item \tb{Init}: $\mc{A}$ initially commits an initial set $\tilde{S}
\subseteq [L]$.

\item \tb{Setup}: $\mc{C}$ obtains public parameters $PP$ by running 
$\tb{Setup}(1^{\lambda}, 1^L)$ and gives $PP$ to $\mc{A}$.

\item \tb{Query Phase}: $\mc{C}$ generates a key pair $(USK_j, UPK_j)$ by
running $\tb{GenKey} (j, PP)$ for all $j \in \tilde{S}$. It gives $\{ (j,
UPK_j) \}_{j \in \tilde{S}}$ to $\mc{A}$.

\item \tb{Challenge}: $\mc{A}$ submits a challenge set $S^* \subseteq 
\tilde{S}$. $\mc{C}$ obtains a ciphertext tuple $(CH^*, CK^*)$ by running
$\tb{Encaps} (S^*, \{ (j, UPK_j) \}_{j \in S^*}, PP)$. It sets $CK_0^* =
CK^*$ and $CK_1^* = RK$ by selecting a random $RK$. It flips a random coin 
$\mu \in \bits$ and gives $(CH^*, CK_\mu^*)$ to $\mc{A}$.

\item \tb{Guess}: Finally, $\mc{A}$ outputs a guess $\mu' \in \bits$, and wins
the game if $\mu = \mu'$.
\end{enumerate}
The advantage of $\mc{A}$ is defined as $\Adv_{DBE,\mc{A}}^{SS} (\lambda) = 
\big| \Pr[\mu = \mu'] - \frac{1}{2} \big|$ where the probability is taken over
all the randomness of the experiment. A DBE scheme is semi-statically secure
if for all probabilistic polynomial-time (PPT) adversary $\mc{A}$, the
advantage of $\mc{A}$ is negligible in the security parameter $\lambda$.
\end{definition}

The adaptive security model is the strongest security model of BE in which an 
attacker can specify the challenge set $S^*$ in the challenge phase
\cite{GentryW09}. In the adaptive security model, a challenger first generates
public parameter $PP$ and gives it to an attacker. Then, the attacker requests
a key generation query for a user index to obtain the user's public key, and a
key reveal query for a user index to obtain the user's private key. In the
challenge phase, the attacker submits the challenge set $S^*$ that does not
include the user index in key exposure queries and obtains the challenge
ciphertext header $CH^*$ and the challenge session key $CK_\mu^*$. Finally,
the attacker succeeds if he can guess whether the challenge session key is
correct or random.

\begin{definition}[Adaptive Security] \label{def:dbe-ad-sec}
The adaptive security of DBE is defined in terms of the following
experiment between a challenger $\mc{C}$ and a PPT adversary $\mc{A}$ where
$1^\lambda$ and $1^L$ are given as input:
\begin{enumerate}
\item \tb{Setup}: $\mc{C}$ obtains public parameters $PP$ by running 
$\tb{Setup}(1^{\lambda}, 1^L)$ and gives $PP$ to $\mc{A}$.

\item \tb{Query Phase}: $\mc{A}$ adaptively requests key generation and key
corruption queries. These queries are processed as follows:
    \begin{itemize}
	\item Key Generation: $\mc{A}$ issues this query on an index $i \in [L]$
	such that $i \not\in KQ$. $\mc{C}$ creates $(USK_i, UPK_i)$ by running 
	$\tb{GenKey}(i, PP)$, adds $i$ to $KQ$, and responds $UPK_i$ to $\mc{A}$.

	\item Key Corruption: $\mc{A}$ issues this query on an index $i \in [L]$
	such that $i \in KQ \setminus CQ$. $\mc{C}$ adds $i$ to $CQ$ and 
	responds $USK_i$ to $\mc{A}$. 
	\end{itemize}

\item \tb{Challenge}: $\mc{A}$ submits a challenge set $S^* \subseteq KQ
\setminus CQ$. $\mc{C}$ obtains a ciphertext tuple $(CH^*, CK^*)$ by running
$\tb{Encaps} (S^*, \{ (j, UPK_j) \}_{j \in S^*}, PP)$.  It sets $CK_0^* =
CK^*$ and $CK_1^* = RK$ by selecting a random $RK$. It flips a random coin 
$\mu \in \bits$ and
gives $(CH^*, CK_\mu^*)$ to $\mc{A}$.

\item \tb{Guess}: Finally, $\mc{A}$ outputs a guess $\mu' \in \bits$, and
wins the game if $\mu = \mu'$.
\end{enumerate}
The advantage of $\mc{A}$ is defined as $\Adv_{DBE,\mc{A}}^{AD} (\lambda) =
\big| \Pr[\mu = \mu'] - \frac{1}{2} \big|$ where the probability is taken over
all the randomness of the experiment. A DBE scheme is adaptively secure if for
all probabilistic polynomial-time (PPT) adversary $\mc{A}$, the advantage of
$\mc{A}$ is negligible in the security parameter $\lambda$.
\end{definition}

\begin{lemma}[\cite{GentryW09,KolonelosMW23}] \label{lem:conv-ss-to-ad}
Let $\Pi_{SS}$ be a semi-statically secure DBE scheme. Then there exists 
$\Pi_{AD}$ that is an adaptively secure DBE scheme.
\end{lemma}

The active-adaptive security model is a modification of the adaptive security
model for DBE to allow the registration of malicious user public keys
\cite{KolonelosMW23}. This active-adaptive security model is very similar to
the adaptive security model above except that an attacker additionally
requests a malicious corruption query to register a malicious user public key.
The detailed description of this security model is given as follows:

\begin{definition}[Active-Adaptive Security] \label{def:dbe-aa-sec}
The active-adaptive security of DBE is defined in terms of the following 
experiment between a challenger $\mc{C}$ and a PPT adversary $\mc{A}$ where
$1^\lambda$ and $1^L$ are given as input:
\begin{enumerate}
\item \tb{Setup}: $\mc{C}$ obtains public parameters $PP$ by running 
$\tb{Setup}(1^{\lambda}, 1^L)$ and gives $PP$ to $\mc{A}$.

\item \tb{Query Phase}: $\mc{A}$ adaptively requests key generation, key
corruption, and malicious corruption queries. These queries are processed as 
follows:
    \begin{itemize}
    \item Key Generation: $\mc{A}$ issues this query on an index $i \in [L]$ 
	such that $i \not\in KQ \wedge i \not\in MQ$. $\mc{C}$ creates $(USK_i, 
	UPK_i)$ by running $\tb{GenKey}(i, PP)$, adds $i$ to $KQ$, and responds 
	$UPK_i$ to $\mc{A}$. 

    \item Key Corruption: $\mc{A}$ issues this query on an index $i \in [L]$
	such that $i \in KQ \wedge i \not\in CQ$. $\mc{C}$ adds $i$ to $CQ$ and 
	responds with $USK_i$ to $\mc{A}$. 

	\item Malicious Corruption: $\mc{A}$ issues this query on an index $i \in 
	[L]$ such that $i \not\in KQ \wedge i \not\in MQ$. $\mc{C}$ adds $i$ to
	$MQ$ and stores $UPK_i$. 
    \end{itemize}

\item \tb{Challenge}: $\mc{A}$ submits a challenge set $S^* \subseteq KQ
\setminus (CQ \cup MQ)$. $\mc{C}$ obtains a ciphertext tuple $(CH^*, CK^*)$
by running $\tb{Encaps} (S^*, \{ (j, UPK_j) \}_{j \in S^*}, PP)$. It sets
$CK_0^* = CK^*$ and $CK_1^* = RK$ by selecting a random $RK$. It flips a 
random coin $\mu \in \bits$ and gives $(CH^*, CK_\mu^*)$ to $\mc{A}$.

\item \tb{Guess}: Finally, $\mc{A}$ outputs a guess $\mu' \in \bits$, and wins 
the game if $\mu = \mu'$.
\end{enumerate}
The advantage of $\mc{A}$ is defined as $\Adv_{DBE,\mc{A}}^{AA} (\lambda) = 
\big| \Pr[\mu = \mu'] - \frac{1}{2} \big|$ where the probability is taken over
all the randomness of the experiment. A DBE scheme is active-adaptively secure
if for all probabilistic polynomial-time (PPT) adversary $\mc{A}$, the
advantage of $\mc{A}$ is negligible in the security parameter $\lambda$.
\end{definition}

\begin{lemma}[\cite{KolonelosMW23}] \label{lem:conv-ad-to-aa}
Let $\Pi_{AD}$ be an adaptively secure DBE scheme. Then $\Pi_{AD}$ is also 
active-adaptively secure.
\end{lemma}

\section{Construction}

In this section, we propose a basic DBE scheme for the semi-static security
and an enhanced DBE scheme for the adaptive security.

\subsection{Semi-Static Construction} \label{sec:dbe-ss-scheme}

Our basic DBE$_{SS}$ scheme has a similar structure to the KMW-DBE scheme with
linear size public parameters, constant size ciphertexts, and constant size
private keys of Kolonelos et al. \cite{KolonelosMW23}. However, our DBE$_{SS}$
scheme uses composite-order bilinear groups instead of prime-order bilinear
groups and modifies some group elements to use the dual system encryption
technique in the security proof. The detailed description of our basic
DBE$_{SS}$ scheme is given as follows:

\begin{description}
\item [$\tb{DBE}_{SS}.\tb{Setup}(1^\lambda, 1^L)$:] Let $\lambda$ be a security
parameter and $L$ be the maximum number of users. It first generates bilinear
groups $\G, \G_T$ of composite order $N = p_1 p_2 p_3$ where $p_1, p_2$, and
$p_3$ are random primes. It selects random generators $g_1, g_3$ of $\G_{p_1},
\G_{p_3}$ respectively.
It selects random $\alpha \in \Z_N$ and $u \in \G_{p_1}$. Next, it selects 
random $\{ Y_k \}_{1 \leq k \leq 2L} \in \G_{p_3}$ and creates $\{ A_k = 
g^{\alpha^k} \}_{1 \leq k \leq L}$, $\{ U_k = u^{\alpha^k} Y_k \}_{1 \leq k 
\leq 2L}$. 
It chooses a pairwise independent hash function $\tb{H}$ such that
$\tb{H} : \G_T \rightarrow \bits^\lambda$. It outputs public parameters 
	\begin{align*}
    PP = \Big(
    &   (N, \G, \G_T, e),~ g = g_1, Y = g_3,~
        \big\{ A_k \big\}_{1 \leq k \leq L},~ 
		\big\{ U_k \big\}_{1 \leq k \neq L+1 \leq 2L},~ 
        \Omega = e(g, U_{L+1}),~ \tb{H}
    \Big).
    \end{align*}

\item [$\tb{DBE}_{SS}.\tb{GenKey}(i, PP)$:] It selects random $\gamma_i \in 
\Z_N$ and $\{ Y_k \}_{1 \leq k \leq L} \in \G_{p_3}$. It outputs a private 
key $USK_i$ and a public key $UPK_i$ as 
	\begin{align*}
    USK_i = \Big( K_i = U_{L+1-i}^{\gamma_i} Y_{L+1-i} \Big),~ 
    UPK_i = \Big( V_i = g^{\gamma_i},~
        \big\{ V_{i,k} = U_{k}^{\gamma_i} Y_k 
		\big\}_{1 \leq k \neq L+1-i \leq L}
    \Big).
    \end{align*}

\item [$\tb{DBE}_{SS}.\tb{IsValid}(j, UPK_j, PP)$:] Let $UPK_j = (V_j, \{
V_{j,k} \})$. It computes $T = e(V_j, U_L)$. For all $k \in \{ 1, \ldots, L
\} \setminus \{ L+1-j \}$, it checks that $T \stackrel{?}{=} e(A_{L-k},
V_{j,k})$ where $A_0 = g$. If it passes all checks, then it outputs 1.
Otherwise, it outputs 0.

\item [$\tb{DBE}_{SS}.\tb{Encaps}(S, \{ (j, UPK_j) \}_{j \in S}, PP)$:] 
Let $UPK_j = (V_j, \{ V_{j,k} \})$. It selects random $t \in \Z_N$ and outputs 
a ciphertext header 
	\begin{align*}
    CH = \Big(
    C_1 = g^t,~
    C_2 = \Big( \prod_{j \in S} A_j V_j \Big)^t
    \Big)
    \end{align*}
and a session key $CK = \tb{H}(\Omega^t)$.

\item [$\tb{DBE}_{SS}.\tb{Decaps}(S, CH, i, USK_i, \{ (j, UPK_j) \}_{j \in S}, 
PP)$:] Let $CT = (C_1, C_2, C)$, $USK_i = K_i$, and $UPK_j = (V_j, \{ 
V_{j,k} \})$. If $i \not\in S$, it outputs $\perp$.
It computes decryption components
	\begin{align*}
	D_1 = K_i,~
	D_2 = U_{L+1-i},~
	D_3 = \prod_{j \in S \setminus \{i\}} U_{L+1-i+j} V_{j,L+1-i}.
	\end{align*}
It outputs a session key 
	$CK = \tb{H}( e(C_2, D_2) \cdot e(C_1, D_1 \cdot D_3)^{-1} )$.
\end{description}

To show the correctness of the basic DBE scheme, we show that a correct session 
key can be derived. If $i \in S$, then we can check that a session element is 
derived by the following equation
    \begin{align*}
    &   e(C_2, D_2) 
     =  e((\prod_{j \in S} A_j V_j)^t, U_{L+1-i}) \\
    &=  e((A_i V_i)^t, U_{L+1-i}) \cdot 
	 	e((\prod_{j \in S \setminus \{i\}} A_j V_j)^t, U_{L+1-i}) \\
    &=  e((g^{\alpha^i})^t, u^{\alpha^{L+1-i}} Y_{}) \cdot
        e((g^{\gamma_i})^t, u^{\alpha^{L+1-i}} Y_{}) \cdot
        e((\prod_{j \in S \setminus \{i\}} g^{\alpha^j} g^{\gamma_j})^t, 
		u^{\alpha^{L+1-i}} Y_{}) \\
    &=  e(g^t, u^{\alpha^{L+1}}) \cdot
        e(g^t, u^{\alpha^{L+1-i} \gamma_i}) \cdot
        e(g^t, \prod_{j \in S \setminus \{i\}} 
			u^{\alpha^{L+1-i+j}} \cdot u^{\alpha^{L+1-i} \gamma_j}) \\
    &=  \Omega^t \cdot e(C_1, K_i) \cdot
        e(C_1, \prod_{j \in S \setminus \{i\}} U_{L+1-i+j} V_{j,L+1-i})
	 = \Omega^t \cdot e(C_1, D_1 \cdot D_3).
    \end{align*}

\subsection{Adaptive Construction} \label{sec:dbe-ad-scheme}

Our enhanced DBE$_{AD}$ scheme is derived by applying the transformation of
Gentry and Waters \cite{GentryW09} to our basic DBE$_{SS}$ scheme. The GW
transformation is a method that transforms a semi-statically secure BE scheme
into an adaptively secure BE scheme and can be applied to DBE schemes as well.
The detailed description of our DBE$_{AD}$ scheme is given as follows:

\begin{description}
\item [$\tb{DBE}_{AD}.\tb{Setup}(1^\lambda, 1^L)$:] Let $\lambda$ be a security
parameter and $L$ be the number of users. It obtains $PP_{SS}$ by running
$\tb{DBE}_{SS}.\tb{Setup}(1^\lambda, 1^{2L})$. It outputs public parameters $PP =
PP_{SS}$.

\item [$\tb{DBE}_{AD}.\tb{GenKey}(i, PP)$:] Let $i \in [L]$. It generates 
key pairs $(USK_{SS,2i}, UPK_{SS,2i})$ and $(USK_{SS,2i-1}, UPK_{SS,2i-1})$ by
running $\tb{DBE}_{SS}.\tb{GenKey}(2i, PP_{SS})$ and $\tb{DBE}_{SS}.\tb{GenKey} 
(2i-1, PP_{SS})$ respectively. It selects a random bit $u \in \bits$ and 
erases $USK_{SS,2i-(1-u)}$ completely. It outputs a private key $USK_i = 
(USK_{SS,2i-u}, u)$ and a public key $UPK_i = (UPK_{SS,2i}, UPK_{SS,2i-1})$.

\item [$\tb{DBE}_{AD}.\tb{IsValid}(j, UPK_j, PP)$:] Let $UPK_j = (UPK_{SS,2j}, 
UPK_{SS,2j-1})$. It checks that $\tb{DBE}_{SS}.\tb{IsValid} (2j, \lb UPK_{SS,2j}, 
PP_{SS}) = 1$ and $\tb{DBE}_{SS}.\tb{IsValid}(2j-1, UPK_{SS,2j-1}, PP_{SS})
= 1$. If it passes all checks, then it outputs 1. Otherwise, it outputs 0.

\item [$\tb{DBE}_{AD}.\tb{Encaps}(S, \{ (j, UPK_j) \}_{j \in S}, PP)$:] 
Let $S \subseteq [L]$ and $UPK_j = (UPK_{SS,2j}, UPK_{SS,2j-1})$. 
	\begin{enumerate}
	\item It selects random bits $z = \{ z_j \}_{j \in S}$ where $z_j \in
	\bits$. Next, it defines two sets $S_0 = \{ 2j - z_j \}_{j \in S}$ and 
	$S_1 = \{ 2j - (1-z_j) \}_{j \in S}$. 

	\item It obtains two ciphertext pairs $(CH_{SS,0}, CK_{SS,0})$ and 
	$(CH_{SS,1}, CK_{SS,1})$ by running $\tb{DBE}_{SS}.\tb{Encaps} \lb(S_0, 
	\{ (k, UPK_{SS,k} \}_{k \in S_0}, PP_{SS})$ and $\tb{DBE}_{SS}.\tb{Encaps} 
	(S_1, \{ (k, UPK_{SS,k} \}_{k \in S_1}, PP_{SS})$ respectively. 

	\item It selects a random message $CK \in \bits^\lambda$. It obtains
	symmetric key ciphertexts $CT_0$ and $CT_1$ by running $\tb{SKE.Encrypt} 
	(CK_{SS,0}, CK)$ and $\tb{SKE.Encrypt} (CK_{SS,1}, CK)$ respectively.

	\item It outputs a ciphertext header $CH = (CH_{SS,0}, CH_{SS,1},
	CT_0, CT_1, z )$ and a session key $CK$. 
	\end{enumerate}

\item [$\tb{DBE}_{AD}.\tb{Decaps}(S, CH, i, USK_i, \{ (j, UPK_j) \}_{j \in S}, 
PP)$:] Let $USK_i = (USK_{SS,2i-u}, u)$. If $i \not\in S$, it outputs 
$\perp$. 
	\begin{enumerate} 
	\item It derives two sets $S_0 = \{ 2j - z_j \}_{j \in S}$ and $S_1 = 
	\{ 2j - (1-z_j) \}_{j \in S}$ If $z_i = u$, then it sets $S' = S_0, 
	CH'_{SS} = CH_{SS,0}, CT' = CT_0$. Otherwise, it sets $S' = S_1, CH'_{SS} 
	= CH_{SS,1}, CT' = CT_1$.
	
	\item It obtains $CK'_{SS}$ by running $\tb{DBE}_{SS}.\tb{Decaps} (S', 
	CH'_{SS}, 2i-u, USK_{SS,2i-u}, \{ (k, UPK_{SS,k} \}_{k \in S'}, \lb
	PP_{SS})$. 

	\item It obtains a decrypted message $CK$ by running $\tb{SKE.Decrypt} 
	(CK'_{SS}, CT')$ and outputs a session key $CK$.
	\end{enumerate}
\end{description}

The correctness of our enhanced DBE scheme easily followed from the
correctness of the underlying SKE and DBE$_{SS}$ schemes.

\section{Security Analysis}

In this section, we show that our DBE$_{SS}$ scheme provides the semi-static
security under static assumptions in composite-order bilinear groups. Then, we
show that our DBE$_{AD}$ scheme provides the adaptive and active-adaptive
security.

We prove the semi-static security of our DBE$_{SS}$ scheme by using the
D{\'{e}}j{\`{a}} Q technique, which is a variant of the dual system encryption
(DSE) technique \cite{Waters09,ChaseM14}. In particularly, we prove our
DBE$_{SS}$ scheme by following the strategy of Wee \cite{Wee16} that was used to
prove the static security of a variant BGW-BE scheme in composite-order
bilinear groups. The basic idea of DSE proof is to change normal ciphertexts
and normal private keys into semi-functional ciphertexts and semi-functional
private keys through hybrid games. In the final game, since the
semi-functional challenge ciphertext and semi-functional private keys are not
related to each other, it is relatively easy for a simulator to generate
semi-functional private keys that is not related to the challenge
semi-functional ciphertext. Thus, it is possible to show that the challenge
session key is random. The D{\'{e}}j{\`{a}} Q technique is very similar to the
DSE technique except that it can be used to change normal private keys to
semi-functional private keys even if the private keys do not have random
variables.

\begin{theorem}[Semi-Static Security] \label{thm:dbe-ss-sec}
The basic DBE scheme is semi-statically secure if the SD and GSD assumptions 
hold.
\end{theorem}

\begin{proof}
We first define the semi-functional type of elements and ciphertext. For
the semi-functional type, we let $g_2$ denote a fixed generator of the
subgroup $\G_{p_2}$.
\begin{description}
\item [\tb{UL-$(\eta,0)$}.] Let $UL = \{ U'_i = u^{\alpha^i} Y_i \}_{i=1}^{2L}$
be a normal list of elements. Let $r_j, a_j$ be fixed random exponents for
index $j \in [k]$. It selects random $Y'_1, \ldots, Y'_{2L} \in \G_{p_3}$ and
outputs a type-$(\eta,0)$ list of elements as $$UL = \Big\{ U_i = U'_i
g_2^{\sum_{j=1}^{k-1} r_j a_j^{L+1-i}} Y'_i \Big\}_{i=1}^{2L}.$$

\item [\tb{UL-$(\eta,1)$}.] Let $UL = \{ U'_i = u^{\alpha^i} Y_i \}_{i=1}^{2L}$
be a normal list of elements. Let $r_j, a_j$ be fixed random exponents for
index $j \in [k]$. It selects random $Y'_1, \ldots, Y'_{2L} \in \G_{p_3}$
outputs a type-$(\eta,1)$ list of elements as $$UL = \Big\{ U_i = U'_i
g_2^{\sum_{j=1}^{k-1} r_j a_j^{L+1-i}} g_2^{r_k \alpha^{L+1-i}} Y'_i
\Big\}_{i=1}^{2L}.$$

\item [\tb{UL-SF}.] Let $UL = \{ U'_i = u^{\alpha^i} Y_i \}_{i=1}^{2L}$ be a
list of normal elements. It chooses random $\delta_1, \ldots, \delta_{2L} \in
\Z_N$, $Y'_1, \ldots, Y'_{2L} \in \G_{p_3}$ and outputs a semi-functional list
of elements as $$UL = \Big\{ U_i = U'_i g_2^{\delta_i} Y'_i
\Big\}_{i=1}^{2L}.$$

\item [\tb{CH-SF}.] Let $CH' = (C'_1, C'_2)$ be a normal ciphertext header.
It chooses random $c, d \in \Z_N$ and outputs a semi-functional ciphertext
header $CH = \big( C_1 = C'_1 g_2^c, C_2 = C'_2 g_2^{c d} \big)$.
\end{description}

The security proof consists of a sequence of hybrid games $\tb{G}_0, \tb{G}_1,
\ldots, \tb{G}_5$. The first game will be the original semi-static 
security game and the last one will be a game in which an adversary has no
advantage. We define the games as follows:
\begin{description}
\item [\tb{Game} $\tb{G}_0$.] This game is the original semi-static security 
game defined in Section \ref{def:dbe-ss-sec}. That is, the simulator of this 
game simply follows the honest algorithms. In this game, all parameters, key 
elements, and the challenge ciphertext are normal.

\item [\tb{Game} $\tb{G}_1$.] This game is almost the same as the game
$\tb{G}_0$ except that the simulator sets $\gamma_i = \gamma'_i - \alpha^i$ by
selecting random $\gamma'_i \in \Z_N$ for each index $i$ and creates the
challenge session key $CK_0^* = \tb{H}(e(g^t, U_{L+1}))$ instead of $CK_0^* = 
\tb{H}(\Omega^t)$ where $U_{L+1} = u^{\alpha^{L+1}} Y_{L+1}$.

\item [\tb{Game} $\tb{G}_2$.] In this game, the challenge ciphertext header is
changed to be semi-functional, but all other elements are still normal.

\item [\tb{Game} $\tb{G}_3$.] Next, we define a new game $\tb{G}_3$. In this
game, we change the distribution of $UL$ from normal to semi-functional.
Because of this change, the public parameters, all key pairs, and the
challenge session key that depend on $UL$ also changed.
For the analysis of this game, we define additional sub-games $\tb{H}_{1,0},
\tb{H}_{1,1}, \ldots, \tb{H}_{\eta,0}, \tb{H}_{\eta,1}, \lb \ldots, \tb{H}_{2L,0},
\tb{H}_{2L,1}, \tb{H}_{2L+1,0}$ that change the type of elements in $UL$ one
by one where $\tb{H}_{1,0} = \tb{G}_2$ and $\tb{H}_{2L+1,0} = \tb{G}_3$. 
A more detailed definition of these sub-games is given as follows:
	\begin{description}
	\item [\tb{Game} $\tb{H}_{\eta,0}$.] This game is similar to the game
	$\tb{G}_2$ excpet that the simulator generates a type-$(\eta,0)$ list $UL$. 

	\item [\tb{Game} $\tb{H}_{\eta,1}$.] This game is also similar to the game
	$\tb{G}_2$ except the simulator generates a type-$(\eta,1)$ list $UL$. 
    \end{description}

\item [\tb{Game} $\tb{G}_4$.] In this game $\tb{G}_4$, the only change from
the game $\tb{G}_3$ is that the simulator generates a semi-functional $UL$. 

\item [\tb{Game} $\tb{G}_5$.] In this final game $\tb{G}_5$, the challenge 
session key $CK_0^*$ is changed to be random. Thus the adversary cannot 
distinguish the challenge session key.
\end{description}
Let $\Adv_{\mc{A}}^{G_j}$ be the advantage of $\mc{A}$ in the game $\tb{G}_j$.
We have that $\Adv_{DBE,\mc{A}}^{SS}(\lambda) = \Adv_{\mc{A}}^{G_0}$, and
$\Adv_{\mc{A}}^{G_5} = 0$. From the following Lemmas \ref{lem:dbe-ss-g0-g1},
\ref{lem:dbe-ss-g1-g2}, \ref{lem:dbe-ss-h0-h1}, \ref{lem:dbe-ss-h1-h2},
\ref{lem:dbe-ss-g3-g4}, and \ref{lem:dbe-ss-g4-g5}, we obtain the equation 
	\begin{align*}
    \Adv_{DBE,\mc{A}}^{SS}(\lambda)
    &\leq \sum_{j=1}^5 \big| \Adv_{\mc{A}}^{G_{j-1}} - \Adv_{\mc{A}}^{G_j} \big|
     \leq \Adv_{\mc{B}}^{SD}(\lambda) + 2L \cdot \Adv_{\mc{B}}^{GSD}(\lambda) +
		O({L^2}/{p_2})
    \end{align*}
where $L$ is the number of users. This completes the proof.
\end{proof}

\begin{lemma} \label{lem:dbe-ss-g0-g1}
No adversary can distinguish $\tb{G}_0$ from $\tb{G}_1$ since two games 
$\tb{G}_0$ and $\tb{G}_1$ are equal.
\end{lemma}

\begin{proof}
Suppose there exists an adversary $\mc{A}$ that distinguishes $\tb{G}_0$ from
$\tb{G}_1$ with a non-negligible advantage. $\mc{B}$ that interacts with 
$\mc{A}$ is described as follows:

\vs\noindent \tb{Init:} $\mc{A}$ submits an initial set $\tilde{S}$. 

\vs\noindent \tb{Setup:} $\mc{B}$ chooses random $\alpha \in \Z_N$, $u \in 
\G_{p_1}$, $\{ Y_k \}_{1 \leq k \leq 2L} \in \G_{p_3}$ and builds $\{ A_k = 
g^{\alpha^k} \}_{1 \leq k \leq L}, UL = \{ U_i = u^{\alpha^i} Y_i \}_{i=1}^{2L}$. 
It publishes 
	$$PP = \big( (N, \G, \G_T, e), g = g_1, Y = g_3, \{ A_k \}_{1 \leq 
	k \leq L}, \{ U_k \}_{1 \leq k \neq L+1 \leq 2L}, \Omega = e(g, U_{L+1}) 
	\big).$$

\svs\noindent \tb{Query Phase:} For each index $i \in \tilde{S}$, $\mc{B}$ 
selects random $\gamma'_i \in \Z_N$, $\{ Y'_{i,k} \}_{1 \leq k \leq L} \in 
\G_{p_3}$ and creates a public key
	\begin{align*}
	UPK_i = \big( 
	V_i = g^{\gamma'_i} A_i^{-1}, \big\{ 
	V_{i,j} = U_{j}^{\gamma'_i} U_{j+i} Y'_{i,j} 
	\big\}_{1 \leq j \neq L+1-i \leq L} \big)
	\end{align*}
by implicitly setting $\gamma_i = \gamma'_i - \alpha^i$. It gives $\{ (j, 
UPK_j) \}_{j \in \tilde{S}}$ to $\mc{A}$.

\svs\noindent \tb{Challenge:} For a challenge set $S^* \subseteq \tilde{S}$, 
$\mc{B}$ selects random $t \in \Z_N$ and creates a challenge ciphertext 
header and a session key
	\begin{align*}
    CH^* = \big( 
	C_1^* = g^t,~ 
	C_2^* = (g^t)^{\sum_{j \in S^*} \gamma'_j} \big),~
    CK^* = \tb{H}(e(g^t, U_{L+1})).
	\end{align*}
It sets $CK_0^* = CK^*$ and $CK_1^* = RK$ by selecting a random $RK$. Next, 
it flips a random coin $\mu \in \bits$ and gives $(CH^*, CK_\mu^*)$ to $\mc{A}$

\svs\noindent \tb{Guess:} $\mc{A}$ outputs a guess $\mu'$. If $\mu = \mu'$,
then $\mc{B}$ outputs 1. Otherwise, it outputs 0.

\vs\noindent This completes our proof.
\end{proof}

\begin{lemma} \label{lem:dbe-ss-g1-g2}
If the SD assumption holds, then no PPT adversary can distinguish $\tb{G}_1$
from $\tb{G}_2$ with a non-negligible advantage.
\end{lemma}

\begin{proof}
Suppose there exists an adversary $\mc{A}$ that distinguishes $\tb{G}_0$ from
$\tb{G}_1$ with a non-negligible advantage. A simulator $\mc{B}$ that solves
the SD assumption using $\mc{A}$ is given: a challenge tuple $D = ((N, \G,
\G_T, e), g_1, g_3)$ and $Z$ where $Z = Z_0 = X_1 \in \G_{p_1}$ or $Z = Z_1 =
X_1 R_1 \in \G_{p_1 p_2}$. The description of $\mc{B}$ that interacts with
$\mc{A}$ is almost the same as that of Lemma \ref{lem:dbe-ss-g0-g1} except
the generation of the challenge ciphertext header. The challenge ciphertext
header is generated as follows:

\svs\noindent \tb{Challenge:} For a challenge set $S^*$, $\mc{B}$ creates a 
challenge ciphertext header and a session key as
	\begin{align*}
    CH^* = \big( 
	C_1^* = Z,~ C_2^* = (Z)^{\sum_{j \in S^*} \gamma'_j} \big),~
    CK^* = \tb{H}(e(Z, U_{L+1}))
	\end{align*}
where $\{ \gamma'_i \}$ are chosen in the key query step. It sets $CK_0^* =
CK^*$ and $CK_1^* = RK$ by selecting a random $RK$. It flips a random coin 
$\mu \in \bits$ and gives $(CH^*, CK_\mu^*)$ to $\mc{A}$.

\vs If $Z = Z_0 = X_1$, then the simulation is the same as $\tb{G}_0$. If $Z
= Z_1 = X_1 R_1$, then it is the same as $\tb{G}_1$ since the challenge
ciphertext header is semi-functional by implicitly setting $c \equiv 
\text{dlog}(R_1) \bmod p_2, d \equiv \sum_{j \in S^*} \gamma'_j \bmod p_2$. 
Note that $d$ is random since $\{ \gamma'_j \}_{j \in S^*}$ modulo $p_2$ are 
not correlated with their values modulo $p_1$ by the Chinese Remainder Theorem 
(CRT). This completes our proof.
\end{proof}

\begin{lemma} \label{lem:dbe-ss-h0-h1}
If the GSD assumption holds, then no PPT adversary can distinguish
$\tb{H}_{\eta,0}$ from $\tb{H}_{\eta,1}$ with a non-negligible advantage.
\end{lemma}

\begin{proof}
Suppose there exists an adversary $\mc{A}$ that distinguishes
$\tb{H}_{\eta,0}$ from $\tb{H}_{\eta,1}$ with a non-negligible advantage. A
simulator $\mc{B}$ that solves the GSD assumption using $\mc{A}$ is given: a
challenge tuple $D = ((N, \G, \G_T, e), \lb g_1, g_3, \lb X_1 R_1, R_2 Y_1)$
and $Z$ where $Z = Z_0 = X_2 Y_2 \in \G_{p_1 p_3}$ or $Z = Z_1 = X_2 R_3 Y_2
\in \G_N$. Then $\mc{B}$ that interacts with $\mc{A}$ is almost similar to
that of Lemma \ref{lem:dbe-ss-g1-g2} except the generation of $PP$ and the
challenge ciphertext. The setup and challenge step is described as follows:

\vs\noindent \tb{Setup:} $\mc{B}$ chooses random $\alpha \in \Z_N$ and 
implicitly sets $u = X_2 \in \G_{p_1}$. It selects random $r_1, \ldots, 
r_{k-1}, a_1, \lb \ldots, a_{k-1} \in \Z_N$ and builds a list of elements
	\begin{align*}
	UL = \big\{ U_i = Z^{\alpha^i} (R_2 Y_1)^{\sum_{j=1}^{k-1} r_j a_j^i} Y'_i 
	\big\}_{i=1}^{2L}
	\end{align*}
by selecting random $\{ Y'_k \}_{1 \leq k \leq 2L} \in \G_{p_3}$.
It publishes $PP = \big( (N, \G, \G_T, e), g = g_1, Y = g_3, \{ A_k = 
g^{\alpha^k} \}_{1 \leq k \leq L}, \lb \{ U_k \}_{1 \leq k \neq L+1 \leq 2L}, 
\Omega = e(g,Z)^{\alpha^{L+1}} \big)$.

\svs\noindent \tb{Challenge:} For a challenge set $S^*$, $\mc{B}$ creates a 
challenge ciphertext header and a session key as
	\begin{align*}
    CH^* = \big( 
	C_1 = X_1 R_1,~ 
	C_2 = (X_1 R_1)^{\sum_{j \in S^*} \gamma'_j} \big),~
    CK^* = \tb{H}(e(X_1 R_1, U_{L+1})).
	\end{align*}
It sets $CK_0^* = CK^*$ and $CK_1^* = RK$ by selecting a random $RK$. It flips 
a random coin $\mu \in \bits$ and gives $(CH^*, CK_\mu^*)$ to $\mc{A}$.

\svs\noindent \tb{Guess:} $\mc{A}$ outputs a guess $\mu'$. If $\mu = \mu'$,
then $\mc{B}$ outputs 1. Otherwise, it outputs 0.

\vs If $Z = Z_0 = X_2 Y_2$, then the simulation is the same as $\tb{H}_{\eta,0}$. 
If $Z = Z_1 = X_2 R_3 Y_2$, then it is the same as $\tb{H}_{\eta,1}$ by 
implicitly setting $r_\eta \equiv \text{dlog}(Z_1) \bmod p_2$. 
This completes our proof.
\end{proof}

\begin{lemma} \label{lem:dbe-ss-h1-h2}
No adversary can distinguish $\tb{H}_{\eta,1}$ from $\tb{H}_{k+1,0}$ since two
games are equal.
\end{lemma}

\begin{proof}
In the game $\tb{H}_{\eta,1}$ of the Lemma \ref{lem:dbe-ss-h0-h1}, the
simulator builds an element $U_i$ of $UL$. If we implicitly sets $X_2 = u, R_2
= g_2^{r''}, R_3 = g_2^{r_\eta}, r_j = r'' r'_j$, and $a_\eta \equiv \alpha
\bmod p_2$, then we can rewrite $U_i$ as 
	\begin{align*}
	U_i 
	&= 	(X_2 R_3 Y_2)^{\alpha^i} (R_2 Y_1)^{\sum_{j=1}^{\eta-1} r'_j a_j^i} Y'_i 
	 = 	X_2^{\alpha^i} R_2^{\sum_{j=1}^{\eta-1} r'_j a_j^i} R_3^{\alpha^i} Y' \\
	&= 	u^{\alpha^i} (g_2^{r''})^{\sum_{j=1}^{\eta-1} r'_j a_j^i}
		(g_2^{r_\eta})^{a_\eta^i} Y'
	 = 	u^{\alpha^i} g_2^{\sum_{j=1}^{\eta-1} r_j a_j^i} g_2^{r_\eta a_\eta^i} Y'
	\end{align*}
since $a_\eta \equiv \alpha \bmod p_2$ is completely hidden by the CRT. 
This completes our proof.
\end{proof}

\begin{lemma} \label{lem:dbe-ss-g3-g4}
No adversary can distinguish $\tb{G}_3$ from $\tb{G}_4$ with a non-negligible 
advantage.
\end{lemma}

\begin{proof}
In the game $\tb{G}_3$, the $G_{p_2}$ parts of all elements in $UL$ can be 
expressed as following matrix equation
	\begin{align*}
	\left( 
	\begin{array}{cccc}
	a_1      & a_2      & \cdots & a_{2L} \\
	a_1^2    & a_2^2    & \cdots & a_{2L}^2 \\
	\vdots   & \vdots   & \ddots & \vdots \\
	a_1^{2L} & a_2^{2L} & \cdots & a_{2L}^{2L} 
	\end{array}
	\right)
	\left(
	\begin{array}{c}
	r_1 \\ r_2 \\ \vdots \\ r_{2L}
	\end{array}
	\right)
	=
	\left(
	\begin{array}{c}
	\delta_1 \\ \delta_2 \\ \vdots \\ \delta_{2L}
	\end{array}
	\right)
	\bmod p_2.
	\end{align*}
Since the left matrix is the Vandermonde matrix, this matrix is invertible if
$a_1, \ldots, a_{2L}$ are distinct that can happen with probability $O(L^2 /
p_2)$. Thus there is one-to-one correspondence between $(r_1, \ldots, r_{2L})$ 
and $(\delta_1, \ldots, \delta_{2L})$. This completes the proof.
\end{proof}

\begin{lemma} \label{lem:dbe-ss-g4-g5}
No adversary can distinguish $\tb{G}_4$ from $\tb{G}_5$ with a non-negligible 
advantage.
\end{lemma}

\begin{proof}
In the simulation of the game $\tb{G}_4$, the simulator generates all public
keys only using $\{ U_i \}_{1 \leq i \neq L+1 \leq 2L}$. That is, the
$\G_{p_2}$ part of $U_{L+1}$ is not revealed. Then the session key $CK_0^*$ 
is writted as
	\begin{align*}
	CK_0^* 
	= \tb{H}( e(C_1^*, U_{L+1}) )
	= \tb{H}( e(C_1^*, u^{\alpha^{L+1}} g_2^{\delta_{L+1}}) )
	= \tb{H}( e(C_1^*, u^{\alpha^{L+1}}) \cdot e(C_1^*, g_2^{\delta_{L+1}}) ).
	\end{align*}
Thus, $CK_0^*$ has additional $\log p_2$ bits of min-entropy from
$\delta_{L+1}$ as long as $\delta_{L+1}$ is not zero. Then, by the leftover
hash lemma, $\tb{H}(e(C_1^*, U_{L+1})$ is uniformly distributed since $\tb{H}$
is a pairwise independent hash function.
\end{proof}

\begin{corollary}[Adaptive Security]
The above DBE$_{AD}$ scheme is adaptively secure if the DBE$_{SS}$ scheme is 
semi-statically secure and the SKE scheme is OMI secure.
\end{corollary}

\noindent The proof of this corollary is easily obtained from the Lemma
\ref{lem:conv-ss-to-ad}.

\begin{corollary}[Active-Adaptive Security]
The above DBE$_{AD}$ scheme is also active-adaptively secure if the DBE$_{AD}$
scheme is adaptively secure.
\end{corollary}

\noindent The proof of this corollary is also easily obtained from the Lemma
\ref{lem:conv-ad-to-aa}.

\section{Conclusion}

In this paper, we proposed a DBE scheme with constant size ciphertexts,
constant size private keys, and linear size public parameters, and proved the
semi-static security under static assumptions in composite-order bilinear
groups. We also showed that our DBE scheme can be converted an adaptively
secure DBE scheme by doubling the ciphertext size using the GW transformation.
Our DBE scheme is the first DBE scheme with linear size public parameters that
is proven under static assumptions instead of the $q$-Type assumption. An
interesting open problem is to convert our DBE scheme in composite-order
bilinear groups to a DBE scheme in prime-order bilinear groups.

\bibliographystyle{plain}
\bibliography{dbe-with-linear-param}

\begin{thebibliography}{10}

\bibitem{AbdallaKN07}
Michel Abdalla, Eike Kiltz, and Gregory Neven.
\newblock Generalized key delegation for hierarchical identity-based
  encryption.
\newblock In Joachim Biskup and Javier Lopez, editors, {\em Computer Security -
  ESORICS 2007}, volume 4734 of {\em Lecture Notes in Computer Science}, pages
  139--154. Springer, 2007.

\bibitem{AttrapadungT24}
Nuttapong Attrapadung and Junichi Tomida.
\newblock A modular approach to registered {ABE} for unbounded predicates.
\newblock In Leonid Reyzin and Douglas Stebila, editors, {\em Advances in
  Cryptology - {CRYPTO} 2024}, volume 14922 of {\em Lecture Notes in Computer
  Science}, pages 280--316. Springer, 2024.

\bibitem{BonehF01}
Dan Boneh and Matthew~K. Franklin.
\newblock Identity-based encryption from the {W}eil pairing.
\newblock In Joe Kilian, editor, {\em Advances in Cryptology - CRYPTO 2001},
  volume 2139 of {\em Lecture Notes in Computer Science}, pages 213--229.
  Springer, 2001.

\bibitem{BonehGW05}
Dan Boneh, Craig Gentry, and Brent Waters.
\newblock Collusion resistant broadcast encryption with short ciphertexts and
  private keys.
\newblock In Victor Shoup, editor, {\em Advances in Cryptology - CRYPTO 2005},
  volume 3621 of {\em Lecture Notes in Computer Science}, pages 258--275.
  Springer, 2005.

\bibitem{BonehZ14}
Dan Boneh and Mark Zhandry.
\newblock Multiparty key exchange, efficient traitor tracing, and more from
  indistinguishability obfuscation.
\newblock In Juan~A. Garay and Rosario Gennaro, editors, {\em Advances in
  Cryptology - {CRYPTO} 2014}, volume 8616 of {\em Lecture Notes in Computer
  Science}, pages 480--499. Springer, 2014.

\bibitem{ChampionHW25}
Jeffrey Champion, Yao{-}Ching Hsieh, and David~J. Wu.
\newblock Registered {ABE} and adaptively-secure broadcast encryption from
  succinct {LWE}.
\newblock Cryptology ePrint Archive, Report 2025/44, 2025.
\newblock \url{https://eprint.iacr.org/2025/044}.

\bibitem{ChampionW24}
Jeffrey Champion and David~J. Wu.
\newblock Distributed broadcast encryption from lattices.
\newblock In Elette Boyle and Mohammad Mahmoody, editors, {\em Theory of
  Cryptography - {TCC} 2024}, volume 15366 of {\em Lecture Notes in Computer
  Science}, pages 156--189. Springer, 2024.

\bibitem{ChaseM14}
Melissa Chase and Sarah Meiklejohn.
\newblock D{\'{e}}j{\`{a}} {Q}: Using dual systems to revisit q-type
  assumptions.
\newblock In Phong~Q. Nguyen and Elisabeth Oswald, editors, {\em Advances in
  Cryptology - {EUROCRYPT} 2014}, volume 8441 of {\em Lecture Notes in Computer
  Science}, pages 622--639. Springer, 2014.

\bibitem{FiatN93}
Amos Fiat and Moni Naor.
\newblock Broadcast encryption.
\newblock In Douglas~R. Stinson, editor, {\em Advances in Cryptology - CRYPTO
  '93}, volume 773 of {\em Lecture Notes in Computer Science}, pages 480--491.
  Springer, 1993.

\bibitem{GargLWW23}
Rachit Garg, George Lu, Brent Waters, and David~J. Wu.
\newblock Realizing flexible broadcast encryption: How to broadcast to a
  public-key directory.
\newblock In Weizhi Meng, Christian~Damsgaard Jensen, Cas Cremers, and Engin
  Kirda, editors, {\em {ACM} Conference on Computer and Communications
  Security, {CCS} 2023}, pages 1093--1107. {ACM}, 2023.

\bibitem{GargLWW24}
Rachit Garg, George Lu, Brent Waters, and David~J. Wu.
\newblock Reducing the {CRS} size in registered {ABE} systems.
\newblock In Leonid Reyzin and Douglas Stebila, editors, {\em Advances in
  Cryptology - {CRYPTO} 2024}, volume 14922 of {\em Lecture Notes in Computer
  Science}, pages 143--177. Springer, 2024.

\bibitem{GargHMR18}
Sanjam Garg, Mohammad Hajiabadi, Mohammad Mahmoody, and Ahmadreza Rahimi.
\newblock Registration-based encryption: Removing private-key generator from
  {IBE}.
\newblock In Amos Beimel and Stefan Dziembowski, editors, {\em Theory of
  Cryptography - {TCC} 2018}, volume 11239 of {\em Lecture Notes in Computer
  Science}, pages 689--718. Springer, 2018.

\bibitem{GargKPW24}
Sanjam Garg, Dimitris Kolonelos, Guru{-}Vamsi Policharla, and Mingyuan Wang.
\newblock Threshold encryption with silent setup.
\newblock In Leonid Reyzin and Douglas Stebila, editors, {\em Advances in
  Cryptology - {CRYPTO} 2024}, volume 14926 of {\em Lecture Notes in Computer
  Science}, pages 352--386. Springer, 2024.

\bibitem{GayKW18}
Romain Gay, Lucas Kowalczyk, and Hoeteck Wee.
\newblock Tight adaptively secure broadcast encryption with short ciphertexts
  and keys.
\newblock In Dario Catalano and Roberto~De Prisco, editors, {\em Security and
  Cryptography for Networks - {SCN} 2018}, volume 11035 of {\em Lecture Notes
  in Computer Science}, pages 123--139. Springer, 2018.

\bibitem{GentryW09}
Craig Gentry and Brent Waters.
\newblock Adaptive security in broadcast encryption systems (with short
  ciphertexts).
\newblock In Antoine Joux, editor, {\em Advances in Cryptology - EUROCRYPT
  2009}, volume 5479 of {\em Lecture Notes in Computer Science}, pages
  171--188. Springer, 2009.

\bibitem{HohenbergerLWW23}
Susan Hohenberger, George Lu, Brent Waters, and David~J. Wu.
\newblock Registered attribute-based encryption.
\newblock In Carmit Hazay and Martijn Stam, editors, {\em Advances in
  Cryptology - {EUROCRYPT} 2023}, volume 14006 of {\em Lecture Notes in
  Computer Science}, pages 511--542. Springer, 2023.

\bibitem{HsiehWW25}
Yao-Ching Hsieh, Brent Waters, and David~J. Wu.
\newblock A generic approach to adaptively-secure broadcast encryption in the
  plain model.
\newblock In Serge Fehr and Pierre-Alain Fouque, editors, {\em Advances in
  Cryptology - {EUROCRYPT} 2025}, volume 15603 of {\em Lecture Notes in
  Computer Science}, pages 336--365. Springer, 2025.

\bibitem{KolonelosMW23}
Dimitris Kolonelos, Giulio Malavolta, and Hoeteck Wee.
\newblock Distributed broadcast encryption from bilinear groups.
\newblock In Jian Guo and Ron Steinfeld, editors, {\em Advances in Cryptology -
  {ASIACRYPT} 2023}, volume 14442 of {\em Lecture Notes in Computer Science},
  pages 407--441. Springer, 2023.

\bibitem{LewkoW10}
Allison~B. Lewko and Brent Waters.
\newblock New techniques for dual system encryption and fully secure {HIBE}
  with short ciphertexts.
\newblock In Daniele Micciancio, editor, {\em Theory of Cryptography - TCC
  2010}, volume 5978 of {\em Lecture Notes in Computer Science}, pages
  455--479. Springer, 2010.

\bibitem{NaorNL01}
Dalit Naor, Moni Naor, and Jeffery Lotspiech.
\newblock Revocation and tracing schemes for stateless receivers.
\newblock In Joe Kilian, editor, {\em Advances in Cryptology - CRYPTO 2001},
  volume 2139 of {\em Lecture Notes in Computer Science}, pages 41--62.
  Springer, 2001.

\bibitem{Waters09}
Brent Waters.
\newblock Dual system encryption: Realizing fully secure {IBE} and {HIBE} under
  simple assumptions.
\newblock In Shai Halevi, editor, {\em Advances in Cryptology - CRYPTO 2009},
  volume 5677 of {\em Lecture Notes in Computer Science}, pages 619--636.
  Springer, 2009.

\bibitem{Wee16}
Hoeteck Wee.
\newblock D{\'{e}}j{\`{a}} {Q}: Encore! un petit {IBE}.
\newblock In Eyal Kushilevitz and Tal Malkin, editors, {\em Theory of
  Cryptography - {TCC} 2016-A}, volume 9563 of {\em Lecture Notes in Computer
  Science}, pages 237--258. Springer, 2016.

\bibitem{WuQZD10}
Qianhong Wu, Bo~Qin, Lei Zhang, and Josep Domingo{-}Ferrer.
\newblock Ad hoc broadcast encryption.
\newblock In Ehab Al{-}Shaer, Angelos~D. Keromytis, and Vitaly Shmatikov,
  editors, {\em {ACM} Conference on Computer and Communications Security, {CCS}
  2010}, pages 741--743. {ACM}, 2010.

\bibitem{ZhuZGQ23}
Ziqi Zhu, Kai Zhang, Junqing Gong, and Haifeng Qian.
\newblock Registered {ABE} via predicate encodings.
\newblock In Jian Guo and Ron Steinfeld, editors, {\em Advances in Cryptology -
  {ASIACRYPT} 2023}, volume 14442 of {\em Lecture Notes in Computer Science},
  pages 66--97. Springer, 2023.

\end{thebibliography}

\end{document}